\documentclass{llncs}
\usepackage[utf8]{inputenc}
\usepackage{amsmath}
\usepackage{amssymb}
\usepackage{graphicx}
\newcommand{\pre}[1]{^\bullet #1}
\newcommand{\post}[1]{#1 ^\bullet}
\newcommand{\conc}{\textbf{co}}
\newcommand{\cnfo}{\mathrel{\natural}}
\DeclareMathOperator{\unf}{\textsc{unf}}

\DeclareMathOperator{\cgs}{\textsc{cgs}}
\title{Asynchronous Games on Petri Nets and ATL}
\author{Federica Adobbati \and Luca Bernardinello \and Lucia Pomello}
\institute{Dipartimento di informatica, sistemistica e
             comunicazione\\
             Università degli studi di Milano - Bicocca\\
             viale Sarca 336 U14\\
             Milano, Italia}
\begin{document}
\maketitle
%
%
\begin{abstract}
We define a game on distributed Petri nets, where several players
interact with each other, and with an environment. The players,
or users, have perfect knowledge of the current state, and pursue
a common goal.
Such goal is expressed by Alternating-time
Temporal Logic (ATL). The users have a winning strategy if they can cooperate to reach their goal, no matter how the environment behaves. We show
that such a game can be translated into a game on concurrent game
structures (introduced in order to give a semantics to ATL).
We compare our game with the game 
on concurrent game structures and discuss the differences 
between the two approaches. 
Finally, we show that, when we consider memoryless strategies 
and a fragment of ATL, 
we can construct a concurrent game structure from the Petri net, 
such that an ATL formula is verified on the net if, and only if, 
it is verified on the game structure.
\end{abstract}
\section{Introduction}
We describe the interaction between a group of \emph{users} and an 
\emph{environment} through a game defined on distributed Petri nets
\cite{BD11}, \cite{GG13}, and played on their unfoldings. 
We assume that the users have full observability on the system, i.e.: they have perfect information on the structure of the system and of its current state.
The users have a common goal, and each of them can control the 
occurrence of a subset of transitions. 
Some of the transitions belong to the environment, and are 
uncontrollable by the users. The environment is indifferent to 
the users' goal.
The users have a winning strategy if they can cooperate to reach 
their goal, no matter how the environment behaves. 

We express the goal of the users with Alternating-time Temporal Logic (ATL). 
This logic was first defined in \cite{AHK02} on 
concurrent game structures. We compare our game with the game 
on concurrent game structures and discuss the differences 
between the two approaches. 

Finally, we show that, when we consider memoryless strategies 
and 
ATL without the \emph{next} operator, 
we can construct a concurrent game structure from the Petri net, 
such that an ATL formula is verified on the net if, and only if, 
it is verified on the game structure.

The paper is structured as follows. 
In the next section, we recall the basic notions about Petri nets: elementary net systems, and their unfoldings, as well as distributed net systems.
In Sect.~\ref{s:asyncgame} we define the game on the unfolding 
of distributed net systems, and
the notions of strategy for a user and of winning strategy for a set of cooperating users.
In Sect.~\ref{sec:atlpn} we propose 
to specify the winning condition by ATL.
After the syntax of ATL formulas,
we define their semantics on the unfolding. 
In Sec.~\ref{sec:reduction} we present the reduction of games on distributed nets into concurrent game structures; namely,
in Sec.~\ref{s:tbags} we recall the definition of turn-based asynchronous game structures, a special case of concurrent game structures, and provide an intuition of ATL semantics on these structures.
In Sec.~\ref{sec:const} we show how to derive a turn-based asynchronous game structure from a Petri net game. 
In Sec.~\ref{sec:relplay} we prove the strict relationship between the plays on the unfolding and the infinite fair computations on the turn-based asynchronous game structure.
The main result is in Sec.~\ref{sec:relstrat}, where we show that a set of memoryless winning strategies exists on the unfolding of a distributed net system for an ATL formula, without next operator, if and only if there exists a set of memoryless winning strategies on the associated game structure, for the same formula. Moreover, in the case of full memory strategies, we prove that if there are winning strategies on the net system, the same holds on the game structure.
Sec.~\ref{sec:last} concludes the paper with a discussion on related works and on some critical points,
indicating possible developments of this work. 
\section{Petri nets}
Petri nets were introduced by C.A. Petri as a formal tool to represent
flows of information in distributed systems. 
In the last decades, several classes of nets have been defined and studied. 
In this work we use the class of \emph{elementary} nets, as defined 
in \cite{RE96}.
\begin{definition} 
A \emph{net} is a triple $N=(P, T, F)$, where $P$ and $T$  are disjoint sets.
The elements of $P$ are called \emph{places} and are represented by circles, 
the elements of $T$ are called \emph{transitions} and are represented by squares.
$F$ is called \emph{flow relation}, with
$F \subseteq (P \times T )\cup (T \times P)$, and is represented by arcs.
\end{definition}
%
For each element of the net $x \in P \cup T$, the \emph{pre-set} of $x$
is the set $\pre{x} = \{ y \in P \cup T \mid (y,x) \in F \}$, 
the \emph{post-set} of $x$ is the set
$\post{x} = \{ y \in P \cup T \mid (x,y) \in F\}$.

We assume that
each transition has non-empty pre-set and post-set: 
$\forall t \in T$, $\pre{t} \neq \emptyset$ and
$\post{t} \neq \emptyset$.

Two transitions, $t_1$ and $t_2$, are \emph{independent}  if
$(\pre{t_1} \cup \post{t_1})$ and $(\pre{t_2} \cup \post{t_2})$
are disjoint. They are in conflict, denoted with $t_1\# t_2$, if $\pre{t_1} \cap \pre{t_2} \neq \emptyset$.

A net $N'= ( P', T', F')$ is a \emph{subnet} of $N = (P, T, F)$ if 
$P' \subseteq P$, $T', \subseteq T$, and 
$F'$ is $F$ restricted to the elements in $N'$.

\begin{definition} 
	An \emph{elementary net system} is a quadruple
	$\Sigma = (P, T, F, m_0)$
	consisting of a finite net $N = (P,T,F)$ and an
	\emph{initial marking}  $m_0\subseteq P$.
	A \emph{marking} is a subset of $P$ and represents a global state.
\end{definition}
A transition $t$ is \emph{enabled} at a marking $m$, denoted $m[t\rangle$,
if    $\pre{t} \subseteq m \land\ \post{t} \cap m = \emptyset$.
 
A transition $t$, enabled at $m$, can  \emph{occur} (or \emph{fire}) 
producing a new marking
$m' = t^\bullet \cup (m \setminus {}^\bullet t)$, denoted $m[t\rangle m'$.
%
A marking $m'$ is reachable from another marking $m$, 
if there is a sequence $ t_1 t_2 \dots t_n$ such that 
$m[t_1 \rangle m_1 [ t_2 \rangle \dots m_{n-1}[ t_n \rangle m'$;
in this case, we write $m' \in [m \rangle$. 
The set of reachable markings
is the set of markings reachable from the initial marking $m_0$,
denoted $[m_0\rangle$.


In a net system, two transitions, $t_1$ and $t_2$, are
\emph{concurrent} at a marking $m$
if they are independent and both enabled at $m$.

An elementary net system is \emph{contact-free} iff $\forall m \in [m_0\rangle, \forall t \in T$ if $\pre{t} \subseteq m$ then  $ \post{t} \cap m = \emptyset$. 
In this paper we consider only contact-free elementary net systems. 

\vspace{\baselineskip}
\noindent
The non sequential behaviour of contact-free elementary net systems can be recorded 
by occurrence nets, which are used to represent by a single object 
the set of potential histories of a net system.
In the following, by $F^*$ we denote the reflexive and transitive closure of $F$. 

Given two elements $x, y \in P \cup T$ we write $x \cnfo y$, if there exist
$t_1, t_2 \in T : t_1 \neq t_2, t_1 F^* x, t_2 F^* y \land
\exists p \in {}^{\bullet}t_1 \cap {}^{\bullet}t_2$.
\begin{definition} 
A net $N =(B, E, F)$ is an \emph{occurrence net} if
\begin{itemize}
	\item for all $b \in B$, $|\pre{b}| \leq 1$
	\item $F^*$ is a partial order on $B \cup E$
	\item for all $x \in B \cup E$, the set
	$\{ y \in B \cup E \mid y F^* x \}$ is finite
	\item for all $x \in B \cup E$, $x \cnfo x$ does not hold
\end{itemize}
\end{definition}
We will say that two elements $x$ and $y$, $x \neq y$,
of $N$ are \emph{concurrent},
and write $x\, \conc\, y$, if they are not ordered by $F^*$,
and $x \cnfo y$ does not hold.

By $\min(N)$ we will denote the set of minimal elements with
respect to the partial order induced by $F^*$.

A \emph{B-cut} of $N$ is a maximal set of pairwise concurrent
elements of $B$. B-cuts represent potential global states
through which a process can go in a history of the system. 
By analogy with net systems, we will sometimes
say that an event $e$ of an occurrence net is \emph{enabled}
at a B-cut $\gamma$, denoted $\gamma[e\rangle$, if $\pre{e} \subseteq \gamma$.
We will denote by $\gamma + e$ the B-cut
$(\gamma\backslash ^\bullet e) \cup e^\bullet$.
A B-cut is a \emph{deadlock cut} if no event is enabled at it.

Let $\Gamma$ be the set of B-cuts of  $N$.
A partial order on $\Gamma$ can be defined as follows: let
$\gamma _1, \gamma _2$ be two B-cuts. We say
$\gamma _1 < \gamma _2$ iff
\begin{enumerate}
	\item $\forall y \in \gamma _2 \exists x \in \gamma _1: xF^*y$
	\item $\forall x \in \gamma _1 \exists y \in \gamma _2: xF^*y$
	\item $\exists x\in \gamma _1$,$\exists y \in \gamma _2 : xF^+y$
\end{enumerate}
In words, $\gamma_1 < \gamma_2$ if any condition
in the second B-cut is or follows a condition of the first B-cut and
any condition in the first B-cut is or comes before a condition of the
second B-cut (and there exists at least one condition coming before).

A sequence of B-cuts, $\gamma _0 \gamma _1 \dots \gamma _i \dots$
is \emph{increasing} if $\gamma_i < \gamma_{i+1}$ for all $i \geq 0$.
A cut $\gamma$ is \emph{compatible} with an increasing 
sequence of B-cuts $\delta$ iff there are two cuts 
$\gamma_i, \gamma_{i+1}\in \delta$ such that $\gamma_i\leq \gamma \leq \gamma_{i+1}$. 

Given an increasing sequence of B-cuts $\delta$, we define a \emph{refinement} 
$\delta'$ of $\delta$ as an increasing sequence of B-cuts such that for each 
$\gamma$ B-cut in $\delta$, $\gamma$ is also a B-cut in $\delta'$. 
A \emph{maximal} refinement $\delta'$ is an increasing sequence of 
cuts such that there is no $\gamma\not\in\delta'$ compatible with 
$\delta'$.

We say that an event $e \in E$ precedes a B-cut $\gamma$,
and write $e < \gamma$,
iff there is $y \in \gamma$ such that $eF^+y$.
In this case, each element of $\gamma$ either follows $e$ or is
concurrent with $e$ in the partial order induced by the occurrence net. 

The next definitions are adapted from~\cite{Ebp91}. 
\begin{definition} 
A \emph{branching process} of a contact-free elementary net system $\Sigma = (P, T, F, m_0)$
is an occurrence net $N = (B, E, F)$, together with a
labelling function $\mu: B \cup E \rightarrow P \cup T$,
such that
\begin{itemize}
	\item $\mu(B) \subseteq P$ and $\mu(E) \subseteq T$
	\item for all $e \in E$, the restriction of $\mu$ to $\pre{e}$
	is a bijection between $\pre{e}$ and $\pre{\mu(e)}$;
	the same holds for $\post{e}$
	\item the restriction of $\mu$ to $\min(N)$ is a bijection
	between $\min(N)$ and $m_0$
	\item for all $e_1, e_2 \in E$, if $\pre{e_1} ={}\pre{e_2}$ and
	$\mu(e_1) = \mu(e_2)$, then $e_1 = e_2$
\end{itemize}
A \emph{run} of $\Sigma$ is a branching process $(N, \mu)$
such that there is no pair of elements $x, y$ in $N$ such that 
$x \cnfo y$.
\end{definition}
For $\gamma$ a B-cut of $N$, the set $\{\mu(b) \mid b \in \gamma\}$
is a reachable marking of $\Sigma$ (\cite{Ebp91}), and we refer to it as the
marking corresponding to $\gamma$.

Let $(N_1, \mu_1)$ and $(N_2, \mu_2)$ be two branching processes of $\Sigma$, 
where $N_i = (B_i, E_i, F_i)$, $i = 1,2$.
We say that $(N_1, \mu_1)$ is a \emph{prefix}
of $(N_2, \mu_2)$ if $N_1$ is a subnet of $N_2$, and 
$\mu_2|_{B_1\cup E_1} =  \mu_1$. 
%
For any contact-free elementary net system $\Sigma$, there exists a unique,
up to isomorphism, maximal branching process of $\Sigma$.
We will call it the \emph{unfolding} of $\Sigma$,
and denote it by $\unf(\Sigma)$ .

A \emph{run} of $\Sigma$ describes a particular
history of $\Sigma$, in which conflicts have been solved.
Any run of $\Sigma$ is a prefix of the unfolding $\unf(\Sigma)$;
we also say that it is a run on $\unf(\Sigma)$.

In this paper we are interested in Petri nets modelling systems
in which several \emph{users} interact with one another, and
with an \emph{environment}. 
Each user controls a subset of transitions, deciding whether to 
fire them or not when they are enabled.

We also assume that choices among transitions are local; this 
means that every choice is completely determined either by the 
environment or by one of the users.

As a formal setting, we refer to the so-called \emph{distributed net systems},
as introduced and studied in \cite{BD11} and in \cite{GG13}.
\begin{definition} 
  A \emph{distributed net system} over a set $L$ of locations
  is an elementary net system $\Sigma = (P, T, F, m_0)$ together with a map
  \[
    \alpha: (P \cup T) \rightarrow L
  \]
  such that for every $p\in P$, $t\in T$, if $p\in\, ^\bullet t$,
  then $\alpha(p) = \alpha(t)$.
\end{definition}
From now on, we will equivalently denote a distributed net system as 
a pair $\langle \Sigma, \alpha \rangle$ or with $\Sigma = (P, T, F, m_0, \alpha)$.

Let $\langle \Sigma, \alpha \rangle$ be a distributed net system, 
we associate every location with one of the agents interacting 
on $\Sigma$. Specifically, if we are considering a system with 
$k$ users, then
$L = \{ \mathrm{env}, u_1, ..., u_k \}$, 
i.e. the distributed net system has $k+1$ locations, 
representing the environment (\emph{env}) and the
$k$ users (\emph{$u_i$}, $i\in\{1,..,k\}$); 
we denote with $T_{u_i}$ the subset of transitions belonging to 
location $u_i$.
We assume that each user controls all transitions in its location.

The notions of unfolding and run apply in the obvious way to
distributed net systems. We will use
$E_{u_i}$ to denote the set of events in the unfolding 
controllable by user $u_i$
(occurrences of transitions that belong to location $u_i$), 
and $E_{nc} = E \setminus{} \cup_{i\in\{1,..,k\}} E_{u_i}$ to denote
uncontrollable events. 
Uncontrollable transitions are meant to represent actions performed
by the \emph{environment}.
\begin{example}\label{e:dcens_unf}
The net in Fig. \ref{fig:net3} is a distributed net system 
with three locations, where two users (represented with different 
tones of grey) interact with an environment (in white in the 
picture). Fig. \ref{fig:unfnet3} represents its unfolding.
\end{example}
\begin{figure}
    \centering
    \includegraphics[width = 0.65\textwidth]{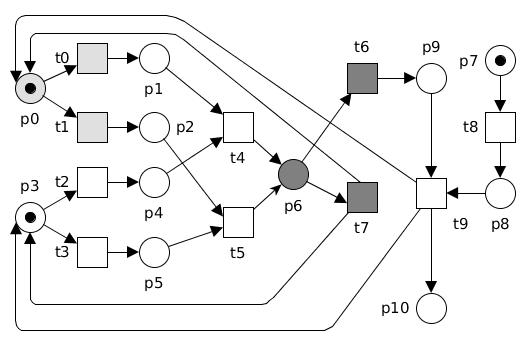}
    \caption{Net with three locations}
    \label{fig:net3}
\end{figure}
\begin{figure}
    \centering
    \includegraphics[width = 0.75\textwidth]{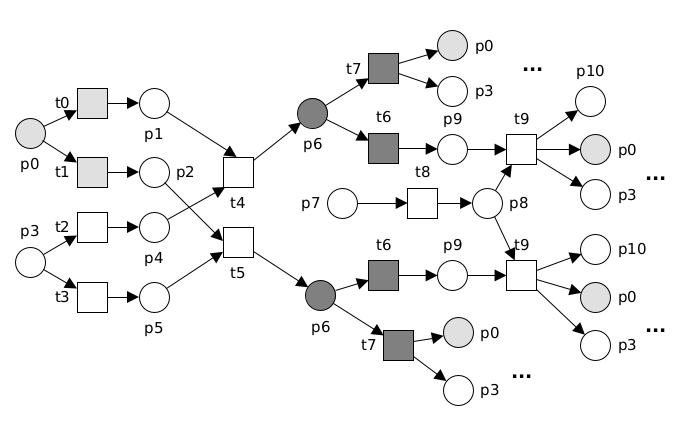}
    \caption{Unfolding of the net in Fig. \ref{fig:net3}}
    \label{fig:unfnet3}
\end{figure}
\section{An asynchronous game played on unfoldings}\label{s:asyncgame}
Let $\Sigma = (P, T, F, m_0, \alpha)$ be a
distributed net system.
We define a game on $\unf(\Sigma)$, adapting some of the ideas 
introduced in \cite{BPPV18} and \cite{ABP19}. 
\begin{definition}
\label{def:play}
Let $\rho = (B_\rho, E_\rho, F_\rho, \mu_\rho)$ be a run on
$\unf(\Sigma$) and $\delta = \gamma_0,\gamma_1,\cdots,$ $\gamma_i,\cdots$
an increasing sequence of B-cuts.
The pair $(\rho, \delta)$ is a \emph{play} iff:
\begin{enumerate}
  \item for each uncontrollable event $e$ in $\unf(\Sigma)$, the net obtained by adding
        $e$ and its postconditions to $\rho$ is not a run of $\unf(\Sigma)$;
  \item if $\rho$ is finite, for each event $e$ in the unfolding controllable by 
  some user, the net obtained by adding
        $e$ and its postconditions to $\rho$ is not a run of $\unf(\Sigma)$;
  \item $\forall e\in E_\rho$ there is a B-cut  $\gamma_i\in\delta$
        such that $e<\gamma_i$. 
\end{enumerate}
\end{definition}
In other words, a play is a run, weakly fair with respect to uncontrollable
transitions, together
with an increasing sequence of B-cuts, which can be seen as a potential
record of the play as observed by an external entity.
In a play, the users have weaker fairness constraints than the 
environment: they are not forced to fire any enabled event, 
if some uncontrollable event is enabled; if the only enabled 
events are controllable by some users, then one of them has to 
fire an event, i.e. a play can be finite only if it ends in a 
deadlock state.
An example of play on the unfolding in Fig. 
\ref{fig:unfnet3} is represented in Fig. \ref{fig:play},
where thick lines show the B-cuts.
\begin{figure}
    \centering
    \includegraphics[width=0.75\textwidth]{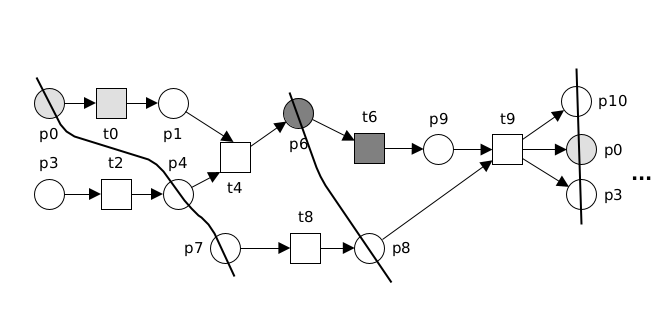}
    \caption{A play on the unfolding}
    \label{fig:play}
\end{figure}
For every pair of cuts $\gamma_i, \gamma_{i+1}$ in the 
play, there may be many events in between, 
that are concurrent or sequential with each other. 

The \emph{winning condition} for a team of users is a set of plays. 
We are particularly interested in the case in which the winning 
plays satisfy a certain property, such as the reachability or 
the avoidance of a target place. 

The behaviour of each user during the play can be determined 
by a \emph{strategy}.
\begin{definition}
\label{def:strategy}
Let $\Gamma$ be the set of B-cuts on the unfolding and $T_u$ be 
the set of transitions controllable by user $u$. 
A \emph{strategy} for a user $u$ is a function 
$f_u:\Gamma\rightarrow 2^{T_u}$ 
such that for every $t\in T_u$ and, for every $\gamma\in \Gamma$, 
if $t\in f_u(\gamma)$, then there is an event $e$ enabled in 
$\gamma$ and such that $\mu(e) = t$.
\end{definition}
A strategy $f_u$ is \emph{memoryless} if for every pair of cuts 
$\gamma_1, \gamma_2$, if $\mu(\gamma_1) = \mu(\gamma_2)$, then
$f_u(\gamma_1) = f_u(\gamma_2)$. 
In this case, we can equivalently define a strategy as a 
function $f_u:Q \rightarrow 2^{T_u}$, where $Q$ is the set of 
reachable markings of the net.

Whenever an event $e\in E_u$ occurs, we say that 
the user $u$ made a \emph{move} in the game. 
\begin{definition}
A user $u$ is \emph{finally postponed} 
in a play $(\rho, \delta)$, 
if there is a cut $\gamma$ in $\delta$ such that 
$f_u(\gamma_j)\neq\emptyset$ for every $\gamma_j\geq \gamma$ compatible with $\delta$, and $u$ did not make any move after 
$\gamma$. 
\end{definition}
From the set of all the plays, we are interested only in those 
in which the users follow their strategy.
\begin{definition}
Let $S\subseteq \{u_1, ..., u_k\}$ be a set of users.
A play $(\rho, \delta)$ is \emph{consistent with a set of 
strategies} $F_S$ iff: 
\begin{enumerate}
    \item For every user $u\in S$, for every event $e\in E_u\cap E_\rho$, 
    there is a cut $\gamma_i\in\delta$ such that 
    $\mu(e)\in f_u(\gamma_i)$, $f_u\in F_S$, and $e$ is the only 
    event between $\gamma_i$ and $\gamma_{i+1}$;
    \item There is no user $u\in S$ finally postponed.
\end{enumerate}
\label{def:cplay}
\end{definition}
A set of strategies $F_S$ is \emph{winning} for the users in 
$S$ iff there is at least a play consistent with $F_S$, and 
the users win every play that is consistent with $F_S$, 
whatever the other agents behave.

We will restrict our attention to the case where all the users 
cooperate to reach the same goal against the environment.
\section{Logical specification of the users' goal}
\label{sec:atlpn}
In Sect.~\ref{s:asyncgame} we have defined the winning condition
as a set of plays. Such a set can be specified in different ways.
In this section, we propose to use Alternating-time Temporal Logic
(ATL), introduced in~\cite{AHK02}.

ATL was introduced as a generalization of CTL, with a more flexible set of path
quantifiers. 
In \cite{AHK02}, ATL is interpreted over \emph{concurrent game 
structures} ($\cgs$), a generalization of Kripke models;
we will formally introduce those structures in Sect.~\ref{s:tbags}.
Intuitively, a $\cgs$ models a system where several \emph{players}
interact. Quantifiers allow to specify that a formula holds if
a subset of players has a strategy which guarantees to reach
a given goal.
In this section we will define a semantics of ATL on the  
unfoldings of distributed net systems.

The elements that characterize an ATL formula are the following: 
a set $P$ of elementary propositions; a finite set of players; 
the symbols $\lor$ and $\lnot$ interpreted in the usual manner; 
path quantifiers $\langle\langle A \rangle\rangle$, where $A$ is
a subset of players, and some temporal operators.

Similarly to LTL and CTL, in ATL the temporal operators are $X$ (next), 
$G$ (always) and $U$ (until). 
In addition to these, we will use the operator $F$ (eventually),
which can be derived from the `until' operator:
for every ATL formula $\phi$, 
$\langle\langle A\rangle\rangle F \phi$ is equivalent to 
$\langle\langle A\rangle\rangle U(\textit{true}, \phi)$.

The set $\Phi$ of ATL formulas is defined as follows:
Let $p \in P$, $\phi_1, \phi_2 \in \Phi$.
\begin{enumerate}
    \item $p$ is an ATL formula;
    \item $\lnot \phi_1$, $\phi_1 \lor \phi_2$ are ATL formulas;
    \item $\langle\langle A\rangle\rangle X \phi_1$,
    $\langle\langle A\rangle\rangle G \phi_1$, 
    $\langle\langle A\rangle\rangle U(\phi_1, \phi_2)$ are ATL formulas, where $A$ is a set of players. 
\end{enumerate}
Let $\Sigma = (P, T, F, m_0, \alpha)$ be a distributed net system.
Let $N =\unf(\Sigma) = (B, E, F, \mu)$ be its unfolding, and $\gamma$ a B-cut of $N$. 

We first define the validity of a formula in a B-cut, and
write
\[
    \gamma\models\phi
\]
to denote that formula $\phi$ holds (or is satisfied) in cut $\gamma$.

Let $\phi_1$ and $\phi_2$ be two ATL formulas and $p\in P$. 
The formula $p$ holds in a B-cut $\gamma$ iff there is a 
$b$ in $\gamma$ such that $\mu(b) = p$.

The formula $\lnot \phi_1$ is true in $\gamma$ iff $\phi_1$ is false in 
$\gamma$;
$\phi_1 \lor \phi_2$ is true in $\gamma$ iff $\phi_1$ is true in $\gamma$ or $\phi_2$ 
is true in $\gamma$. 

Consider now $S\subseteq \{u_1, ..., u_k\}$:
\begin{itemize}
    \item $\langle\langle S \rangle\rangle X \phi_1$ is satisfied 
    in a B-cut $\gamma$ iff there exists a set $F_S$ of 
    strategies, one for each player in $S$, such that for 
    every event $e$ that can occur in $\gamma$, consistent 
    with the strategies, the cut $\gamma + e$ satisfies $\phi_1$.
    \item $\langle\langle S\rangle\rangle G \phi_1$ is satisfied 
    in a B-cut $\gamma$ iff there exists a set $F_S$ of 
    strategies, one for each player in $S$, such that for every 
    play $(\rho', \delta')$ that starts in $\gamma$, consistent 
    with the strategies, in every cut $\gamma'$, compatible 
    with $\delta'$, $\phi_1$ holds.
    \item $\langle\langle S\rangle\rangle U(\phi_2, \phi_1)$ is satisfied 
    in a B-cut $\gamma$ iff there exists a set $F_S$ of 
    strategies, one for each player in $S$, such that for every 
    play $(\rho', \delta')$ that starts in $\gamma$,
    consistent with the strategies:
    \begin{itemize}
        \item  there is a cut $\gamma'$ compatible with 
          $\delta'$ such that $\phi_1$ is satisfied in $\gamma'$;
        \item for every B-cut $\gamma''< \gamma'$ compatible with 
          $\delta'$ , $\phi_2$ is satisfied.
    \end{itemize}
\end{itemize}
An ATL formula is satisfied on $N$ iff it is satisfied in the initial
cut $\gamma_0$ of $N$.
If a formula $\phi$ is satisfied on $N$, we write
\[
    N \models \phi
\]
For an ATL formula in the form $\langle\langle S \rangle\rangle \eta$, 
given a set of strategies $F_S$, we can define a 
\emph{winning play} for the players in $S$. 
Let $(\rho, \delta)$ be a play on the unfolding, consistent with $F_S$.  
\begin{itemize}
    \item  Given an ATL formula $\langle\langle S\rangle\rangle X \phi_1$, 
    $(\rho, \delta)$ is a winning play 
    iff for every cut $\gamma$ reachable from $\gamma_0$ after the 
    occurrence of a single event and compatible with $\delta$, 
    $\gamma$ satisfies $\phi_1$.
    \item Given an ATL formula $\langle\langle S\rangle\rangle G \phi_1$, 
    $(\rho, \delta)$ is a winning play  
    iff every cut compatible with $\delta$ satisfies $\phi_1$.
    \item Given an ATL formula $\langle\langle S\rangle\rangle U(\phi_1, \phi_2)$, 
    $(\rho, \delta)$ is a winning play  
    iff for every $\delta'$ refinement of $\delta$ there is a cut 
    $\gamma$ compatible with $\delta'$ in which $\phi_2$ is satisfied, 
    and in all the cuts $\gamma' < \gamma$ such that $\gamma'$ is 
    compatible with $\delta'$, $\phi_1$ is satisfied.
\end{itemize}
\begin{example}
Consider the net in Fig. \ref{fig:net3} and its 
unfolding in Fig. \ref{fig:unfnet3}. We want to verify 
whether the two users controlling the gray components  
can force the system to reach $p_6$ infinitely often. 
We identify with $u_1$ the user on the light gray component 
and with $u_2$ the user on the dark gray. 
We can express this goal with the formula:
\begin{equation}
\label{eq:reach}
    \langle\langle u_1,u_2\rangle\rangle G \langle\langle u_1,u_2 \rangle\rangle F p_6.
\end{equation}
This formula is satisfied on the unfolding. 
To see that, consider the strategy $f_{u_1}: \Gamma \rightarrow 2^{T_{u_1}}$ 
defined in this way: for every cut $\gamma$ such that 
$\mu(\gamma) = \{p_0, p_4, p_8\}$, $f_{u_1}(\gamma) = \{t_0\}$, 
for every cut $\gamma$ such that 
$\mu(\gamma) = \{p_0, p_5, p_8\}$, $f_{u_1}(\gamma) = \{t_1\}$, 
otherwise $f_{u_1}(\gamma) = \emptyset$; 
and the strategy $f_{u_2}$ defined as $f_{u_2}(\gamma) = \{t_6, t_7\}$ if $\mu(\gamma) = \{p_6, p_8\}$, $f_{u_2}(\gamma) = 
\{t_7\}$ if $\mu(\gamma) = \{p_6, p_{10}\}$, $f_{u_2}(\gamma) = \emptyset$ 
otherwise. 
$F_{u_1, u_2} = \{f_{u_2}, f_{u_1}\}$ is a set of winning strategies 
for the users. 
Indeed, the users have a strategy to reach always cuts 
in which they have a strategy to reach a place labelled as 
$p_6$. 

We could also consider a safety goal, for example we could 
ask that the play never reaches a state in which $p_{10}$ is 
true. This can be expressed with the formula:
\begin{equation}
\label{eq:safe}
    \langle\langle u_1, u_2\rangle\rangle G \lnot p_{10}.
\end{equation}
Since the users must respect a fairness constraint, even if 
this is weaker than the one imposed to the environment, 
they cannot just decide not to fire any transition and 
avoid in this way the unsafe place. 
Nevertheless, they have winning strategies: 
$f'_{u_1}(\gamma) = \{t_1, t_2\}$ if $\mu(\gamma) \in
\{\{p_0, p_4, p_8\}, \{p_0, p_5, p_8\}\}$, 
$f'_{u_1}(\gamma ) = \emptyset$ otherwise; 
$f'_{u_2}(\gamma) = \{t_7\}$ if $\mu(\gamma) = \{p_6, p_8\}$, 
$f'_{u_2}(\gamma) = \emptyset$ otherwise. 
\end{example}
\section{Reduction of games on Petri nets to concurrent game structures}
\label{sec:reduction}
In \cite{AHK02}, the authors present many model checking 
algorithms to decide whether an ATL formula is satisfied on a 
concurrent game structure. In many cases those algorithms can 
be exploited also for the game on Petri nets. 
In particular, we show that when we consider memoryless strategies 
and ATL formulas without the \emph{next (X)} operator, we can 
reduce a game on a net into a turn-based asynchronous game structure 
with fairness constraints.
\subsection{Turn-based asynchronous game structures}
\label{s:tbags}
We first recall from \cite{AHK02} the definition of turn-based asynchronous
game structure, which is a special case of a concurrent
game structure.
\begin{definition}
A \emph{turn-based asynchronous game structure} is a tuple $G = \langle n, Q, \Pi, w, d, \tau\rangle$, where
\begin{itemize}
    \item $n\geq 2$ is the number of \emph{players}. Every player is identified with a number 1,..., n. Player $n$ 
    represents the \emph{scheduler}. At every turn, the 
    scheduler selects one of the other players. 
    \item $Q$ is a finite set of \emph{states}.
    \item $\Pi$ is a set of \emph{propositions}.
    \item $\forall q\in Q$, $w(q)\subseteq \Pi$ is the set of 
    propositions that are true in $q$.
    \item $\forall a\in \{1,..., n\}, q\in Q$, 
    $d_a(q)\in \mathbf{N}$ is the number of moves 
    available for the player $a$ in the state $q$. 
    Every move is identified with a number $1,..., d_a(q)$. 
    For every state $q\in Q$, $d_n(q) = n-1$. 
    \item  For every state $q\in Q$, $D(q)$ is the set 
    $\{1,..., d_1(q)\}\times ... \times\{1,..., d_n(q)\}$ 
    of \emph{move vectors}. 
    \item $\tau$ is the \emph{transition function}. 
    For every $q\in Q$ and $\langle j_1, ..., j_n\rangle$ 
    vector move, $\tau(q, j_1, ..., j_n)\in Q$ is 
    the state where the system is if from the state 
    $q$, every player $a\in\{1, ..., n\}$ chooses move $j_a$. 
    For all move vectors $\langle j_1,..., j_n \rangle$, 
    $\langle j'_1,..., j'_n \rangle$, if $j_n = j'_n = a$ and 
    $j_a = j'_a$, then $\tau(q, j_1,...,j_n) = \tau(q, j'_1,...,j'_n)$, for every $q\in Q$.
\end{itemize}
\end{definition}
Given an initial state, an infinite computation $\lambda$ on a 
turn-based asynchronous game structure is an infinite sequence of states such that 
the successor of each element is fully determined by the moves 
chosen in the previous state.

We denote with $Q^+$ the set of all the finite prefixes of 
the computations $\lambda$ in the concurrent game structure and with $\lambda_i$ 
the prefix of $\lambda$ such that $|\lambda_i| = i$.
A \emph{strategy} for a player $a$ is a function 
$f_a: Q^+ \rightarrow \mathbb{N}$ 
such that $f_a(\lambda_i) \leq d_a(q)$, where $q$ is the last 
element of $\lambda_i$. 
The strategy is \emph{memoryless} if, and only if, for every pair 
of prefixes $\lambda_i, \lambda_j$ ending with the same state, 
$f_a(\lambda_i) = f_a(\lambda_j)$; a memoryless strategy 
can be equivalently defined as a function $f: Q \rightarrow \mathbb{N}$.
We say that a computation $\lambda$ follows a strategy $f_a$ 
iff, for every prefix $\lambda_i$ of a computation, player $a$ 
chooses the move $f_a(\lambda_i)$. 
Let $A$ be a set of players, and $F_A$ be a set of 
strategies, one for each player in $A$.
We denote with \emph{out}$(q, F_A)$ the set of computations 
starting from $q$ and such that the players in $A$ follow the 
strategies in $F_A$. 

On turn-based asynchronous game structures, the syntax of ATL 
is equivalent to the one defined in Sect. \ref{sec:atlpn} for 
distributed net systems. We provide an intuition of the semantics, the formal 
definition is in \cite{AHK02}. 

Let $\langle\langle A \rangle\rangle\eta$ be an ATL formula. 
We can evaluate it by considering a game in which players in $A$ 
are against the others. An infinite computation $\lambda$ is 
winning for players in $A$ iff the computation satisfies the 
formula $\eta$, read as a linear temporal formula, with outermost 
operator $X, G$ or $U$.
The formula $\langle\langle A \rangle\rangle\eta$ is satisfied in a 
state $q$ of the game structure iff there exists a set of 
strategies $F_A$ such that the players in $A$ win all the 
computations in \emph{out}$(q, F_A)$. 
In this case, we say that $F_A$ is a \emph{winning} set of strategies.

We can include some \emph{fairness constraints} to the game 
structure, in order to ignore some computations. 
\begin{definition}
A \emph{fairness constraint} is a pair $\langle a, c \rangle$, 
where $a$ is a player, and $c$ is a function that for every 
state $q\in Q$ selects a subset of moves available for $a$ in $q$.
\end{definition} 
Let $\lambda$ be an infinite computation and $q_i$ its 
$i-th$ element.
A fairness constraint $\langle a, c\rangle$ is \emph{enabled} 
in $q_i$ if $c(q_i) \neq \emptyset$; 
$\langle a, c\rangle$ is \emph{taken} in $q_i$ if there is 
a vector move $\langle j_1, ..., j_n \rangle$ with $j_a \in c(q_i)$ and
$q_{i+1} = \tau(q_i, \langle j_1, ..., j_n \rangle)$.
If $a< n$ we also require that $j_n = a$. 

A computation $\lambda$ is \emph{weakly fair} with respect to 
a fairness constraint $\langle a, c \rangle$ if $\langle a, c \rangle$ 
is not enabled in infinitely many positions of $\lambda$ or if 
it is taken infinitely many times in $\lambda$.
\subsection{Construction of a turn-based asynchronous game
structure from a Petri net game}
\label{sec:const}
Given a distributed net system $\Sigma$ with $k+1$ locations, we construct 
a turn-based asynchronous game structure $G_\Sigma$ in this way:
\begin{itemize}
    \item Every location is represented by a player 
    in $G_\Sigma$. 
    For every $i\in\{1,..., k\}$, we identify the 
    user $u_i$ on the net as player $i$, and the 
    environment as player $k+1$.
    In addition, we insert a fictitious player 
    identified with $k+2$ that has the role of the 
    scheduler.
    \item The reachable markings are all and only 
    the states in $G_\Sigma$. 
    We denote the set of all the states with $Q$, and 
    with $q_0$ the initial marking of $\Sigma$ 
    and the corresponding state in $G_\Sigma$.
    \item We identify the set of propositions as the set $P$ of places: every place can be interpreted as a proposition, 
    that is true if the place is part of the current 
    marking, false otherwise, 
    \item According to the previous point, for every 
    state $q\in Q$, the set of propositions that 
    are true in $q$ is $w(q) = \{p\in P: p\in q\}$.
    \item On the net, for every marking $q$, for every user 
    $a$, $a$ can decide 
    to fire one of the transitions in $T_a$ that are 
    enabled in $q$ or not to move. 
    Then, denoting with $r(q, a)$ the number of transitions 
    that are enabled in $q$ and belong to $T_a$, 
    the player $a$ has $r(q, a)+1$ different possible 
    behaviours in every cut associated with the marking 
    $q$. 
    This is translated on $G_\Sigma$ by putting 
    $d_a(q) = r(q, a)+1$ for every $q\in Q$ and 
    $a$ user. 
    
    $d_{k+1}(q)$ is equal to the number of  
    uncontrollable transitions enabled in $q$. 
    If there are no uncontrollable transitions enabled 
    in $q$, then $d_{k+1}(q) = 1$ and the move 
    remains in the same state.

    For the sake of simplicity, we will refer to each move that 
    changes the state of the system with the name of its associated 
    transition on the distributed net system; we will refer to 
    the move that leaves the system in the same state with $\emptyset$.
    
    Finally, $d_{k+2}(q) = k+1$ for every $q\in Q$. 
    
    The vector $\langle q, j_1, ..., j_{k+2} \rangle$ 
    collects a move of the players in $q$.
    \item For every $q\in Q$, $\tau(q,j_1,...j_{k+2})$ 
    is the state corresponding to the marking that is 
    reached if, starting from the state $q$ in 
    the net, we execute the transition in $j_{j_{k+2}}$.  
\end{itemize}
We represent the fairness constraints on the net with 
weak fairness constraints on $G_\Sigma$. 
We ask that there is no user with a finally postponed move by 
asking that the scheduler is fair, and no player is neglected 
forever. 
Formally, we ask that $\langle k+2, c_j \rangle$ is a weak 
fairness constraint for every $j\in\{1, ..., k+1\}$ and $c_j$ 
is a function such that for every $q\in Q$, $c_j(q) = \{j\}$.

In addition, we need to add constraints to guarantee that 
the plays are weakly fair with respect to uncontrollable transitions. 
For every state $q$, 
if $q$ enables uncontrollable transitions on the net, 
then for every transition $t$ enabled in $q$ we consider the 
subset $t\#(q)$ of uncontrollable transitions in conflict with $t$ 
and enabled in $q$,
and we define a weak constraint in the state $q$ 
in this way: $\langle k+1, c_{t\#} \rangle$, where  
$c_{t\#}(q) = t\#(q) \cup\{t\}$.

Finally, we need to guarantee that, if we are not in a deadlock 
state, but no uncontrollable transition is enabled, the users 
cannot block the system. 
In order to address this problem, we add a set of 
fairness constraints for the users in this way:
for each $q$, in which only users' transitions are 
enabled, and for each player $a_1,..., a_l$ with enabled 
transitions in $q$, we define the fairness constraint 
$\langle a_i, c_q\rangle$, with $i\in \{1,...,l\}$, such that 
$c_q(q)$ is the set of all the moves that do not keep 
the concurrent game structure in the same state, 
and $c_q(q') = \emptyset$ for every $q'\neq q$. 

In the special case of a single user playing against the 
environment, we can avoid this last 
set of constraints and just remove the possibility for the 
user to stay in the same state, if only controllable 
transitions are enabled in that state.
\subsection{Relation between the plays in the two models} 
\label{sec:relplay}
%
Let $\Sigma = (P, T, F, q_0, \alpha)$ be a distributed net
system, and $G_\Sigma$ the associated game structure, as defined
in Sect.~\ref{sec:const}.

We now prove some propositions that show the relations between 
the plays on the unfolding and the infinite fair computations 
on the concurrent game structure. This will be helpful to find 
the relation between the existence of winning strategies in the 
two models.

Let $\lambda=q_0q_1\cdots$ be a fair computation on $G_\Sigma$.
For every pair of consecutive states $q_i, q_{i+1}$, such that 
$q_i \neq q_{i+1}$, the move on the concurrent game structure is
associated, by construction, with an enabled transition in $\Sigma$.

Given $\lambda$, we construct a run $\rho$ on the unfolding
of $\Sigma$, starting 
from the initial cut and adding the events associated 
with the moves that bring from a state to the next one in the same 
order that the states have in $\lambda$.
At the same time we 
construct an increasing sequence $\delta$ of cuts: initially, $\delta = \gamma_0$; 
after an event $e$ is added to the run, we add to $\delta$ the cut 
$\gamma = \gamma'\setminus {\pre{e}} \cup \post{e}$, where 
$\gamma'$ is the last cut added in $\delta$ before $e$.
The pair derived in this way will be denoted
by $(\rho,\delta)[\lambda]$.
\begin{proposition}
\label{prop:pcgs2pn}
The pair $(\rho,\delta)[\lambda]$ is a play on $\unf(\Sigma)$.
\end{proposition}
\begin{proof}
The sequence of cuts $\delta$ satisfies condition 3 
of Definition \ref{def:play} by construction. 

We need to prove that there is no uncontrollable event $e\not\in\rho$ 
such that $\rho \cup \{e\}$ is a run on the unfolding. 
By contradiction, assume that we can find such an event $e$. 
Then, there is $n\in \mathbb{N}$ such that, on $G_\Sigma$, the move $\mu(e) = t$ is 
enabled in every state $q_m\in \lambda : m>n$, but is never chosen. 
The reason cannot be that the environment player is never selected 
by the scheduler, otherwise, the constraint $\langle k+2, c_{k+1}\rangle$ 
would not be satisfied, and $\lambda$ would not be a fair computation. 
Hence the move $t$ is available in every state $q_m$ for the player $k+1$, but he never chooses it. 
Then the computation does not respect the fairness constraint 
$\langle k+1, c_{t\#} \rangle$, where, for each $q\in Q$, 
$c_{t\#}(q)$ is the set $\{t\} \cup t\#(q)$: by construction, if a move 
associated with $t$ or with one of the transitions in conflict with $t$ 
occurs in $q_m$, then $t$ would not be available in $q_{m+1}$. 

Finally, we have to show that point 2 of definition \ref{def:play} holds: 
by contradiction, we assume that $\rho$ is finite, and there 
is an event $e\in E_a$, controllable by user $a$, such that $\rho \cup \{e\}$ is a run. 
If $\rho$ is finite, there is a position $i$ in $\lambda$ such 
that $q_j = q_{j+1}$ for every $j\geq i$. In the state $q_j$, 
there is no uncontrollable enabled transition, otherwise 
point 1 of the definition would not be respected. 
Then there must be a fairness constraint 
$\langle a, c_{q_j} \rangle$ not respected by $\lambda$. 
Hence, if $\lambda$ is fair, there cannot be controllable enabled 
transitions in $q_j$, and $q_j$ must be a deadlock in the distributed 
net system.

Hence $(\rho, \delta)[\lambda]$ constructed in this way is a play
on the unfolding.
\qed\end{proof}
Given a computation $\lambda = q_0q_1...q_n...$ (finite or infinite) 
we can construct a new sequence $\pi(\lambda)$ in which all the 
consecutive states $q_i = q_{i+1}$ has been identified. 
Given an infinite sequence $\lambda$, $\pi(\lambda)$ can be also 
infinite, if there is no state $q_i\in \lambda$ such that for each 
$q_j : j > i$ $q_j = q_{j+1}$, or finite otherwise.

In general, a game on the unfolding is not equivalent 
to a game on the concurrent game structure. This is mainly 
due to the fact that infinite sequences of states on 
the asynchronous game structure are not equivalent to sequences of cuts on 
the unfolding: (1) the cuts do not report the presence of 
moves that do not change the state of the system, 
while the sequences of states do; (2) on the unfolding of 
distributed net systems, 
the memory is `stored' in the cuts, from which a user can 
determine which set of events fired, but not their total 
order, while a sequence of events 
gives full information about the order. 

In order to address point (1), 
we consider only ATL
formulas that do not use the $X$ (next) operator. This is due to the 
fact that, given two infinite computations $\lambda_1$ and $\lambda_2$, 
if $\pi(\lambda_1) = \pi(\lambda_2)$, then the players cannot 
distinguish on the system whether $\lambda_1$ or $\lambda_2$ occurred, 
hence we want that $\lambda_1$ is a winning play for the users iff $\lambda_2$ is.
\begin{lemma}
\label{prop: equiv}
Let $\Sigma$ be a distributed net system, $\psi$ be an ATL formula that does 
not use the operator $X$ and $\lambda_1$ and $\lambda_2$ two infinite 
weakly fair computations on $G_\Sigma$ such that 
$\pi(\lambda_1) = \pi(\lambda_2)$. 
Then if there is a winning strategy for the users, $\lambda_1$ is a 
winning play for the users iff $\lambda_2$ is.
\end{lemma}
\begin{proof}
If the operator $X$ is not allowed, then the
validity of $\psi$ depends only on the sequence of distinct states
in computations; since, by hypothesis, $\pi(\lambda_1) = \pi(\lambda_2)$,
the thesis follows.
\qed\end{proof} 
Every play on the unfolding can be associated with a set 
of infinite fair computations on the turn-based asynchronous game structure: 
let $(\rho, \delta)$ be a play on the unfolding; 
as first step, we consider all the possible sequentializations 
included between the initial cut $\gamma_0$ and the next 
cut $\gamma_1\in \delta$. For each of these linearizations we can 
find a prefix of a computation on the concurrent game structure by executing 
on it, from the state corresponding to the initial marking, 
all the events in the order given by the linearization; 
then, we extend all these prefixes by considering the 
successive pairs of consecutive cuts, and the sequentializations 
of the events between them. 
If the play on the distributed net system reaches a deadlock, then the 
only possibility in the associated concurrent game structure is to execute the 
transition that remains in the same state infinitely often. 

In this way, we obtain a set of infinite computations 
associated with a play $(\rho, \delta)$, 
denoted by $\Lambda(\rho, \delta)$. 
\begin{proposition}
Let $(\rho, \delta)$ be a play on $\unf(\Sigma)$.
For every computation $\lambda\in \Lambda(\rho, \delta)$ on $G_\Sigma$, 
there is at least a computation $\lambda': \pi(\lambda') = \pi(\lambda)$ 
satisfying the fairness constraints defined in Section \ref{sec:const}.
\label{prop:ppn2cgs}
\end{proposition}
\begin{proof}
The first set of constraints guarantees that no player is neglected 
forever by the scheduler. 
If $\rho$ ends with a deadlock, then 
it is easy to see that every associated computation is 
fair on $G_\Sigma$, because the only available move for 
all players after a finite number of steps is the one that 
remains in the same state. 

For the same reason, when $\rho$ is finite, also the other constraints 
are respected: after a finite number of states there is 
no enabled transition, therefore those 
constraints are always taken when the system reaches 
the state corresponding to the deadlock.

We now consider the case in which $\rho$ is infinite. 
Since a play on the unfolding must be fair with respect to 
uncontrollable transitions, in the computation $\lambda$ 
the fairness constraint $\langle k+2, c_{k+1}\rangle$ must be respected. 
If $\rho$ has infinitely many events belonging to location
$u_i$, then the fairness constraint $\langle k+2, c_i\rangle$ 
is also satisfied by construction; 
otherwise, since in case of an infinite 
play, the users have always 
the possibility not to move in the system, we can 
construct a sequence  $\lambda' : \pi(\lambda) = \pi(\lambda')$ 
with some repeated states, that represent points of the sequence in which the 
user was selected by the scheduler in $G_\Sigma$, but the 
state of the system did not change. 

The set of fairness constraints for the environment $k+1$ must be 
satisfied: by contradiction, we assume that there is a 
function $c_{t\#}$, such that $\langle k+1, c_{t\#}\rangle$ 
is not satisfied in $\lambda'$; this means that $\langle k+1, c_{t\#}\rangle$ 
is not enabled in a finite number of states in $\lambda$, but 
it is taken only a finite number of times. 
This means that there is an uncontrollable event $e\not\in\rho$ in the 
net that is enabled in all the cuts compatible with $\delta$, except 
for a finite set; since $\rho$ is infinite, there must be a 
cut $\gamma\in \delta$ such that $e$ is enabled in $\gamma$ and 
in all the cuts $\gamma_j \in \rho : \gamma_j > \gamma$; hence 
$\rho \cup \{e\}$ is a run on $\unf(\Sigma)$.
This is in contradiction with the hypothesis that $\rho$ is a play. 

Finally, the set of constraints for the users are also satisfied: 
consider a fairness constraint referred to a user $a$ in a state 
$q$; since there is no position $i$ in $\lambda$ such that 
$q_j = q$ for every $j\geq i$, there are infinitely many 
positions in which the constraint is not enabled. 
\qed\end{proof}
%
%
\begin{lemma}
\label{lem:eqwc}
Let $\psi \equiv \langle\langle users \rangle\rangle \eta$ be an ATL formula where
each path quantifier is the set of all the users and without the $X$ operator:
\begin{enumerate}
    \item if $(\rho, \delta)$ is a winning play for the users on 
    the unfolding, then any infinite fair computation $\lambda\in 
    \Lambda(\rho, \delta)$ is winning for the users;
    \item if $\lambda$ is an infinite fair computation starting in 
    $q_0$ that is winning for the users, 
    the derived play $(\rho, \delta)[\lambda]$ 
    on the unfolding is winning for the users.
\end{enumerate}
\end{lemma}
\begin{proof}
In order to prove this lemma, we have to consider $G$ and $U$ as 
outermost operator in $\eta$.
\begin{itemize}
    \item Let $\psi\equiv\langle\langle users\rangle\rangle G \phi$,  
    $(\rho, \delta)$ a winning play for the users and $\lambda\in \Lambda(\rho, \delta)$. 
    By construction, $\lambda$ is a sequence of states associated 
    with the markings of a refinement $\delta'$ of $\delta$, 
    in the same order as in $\delta'$. 
    Since $\phi$ must be satisfied in every cut $\gamma\in \delta'$, 
    and for every $\gamma$ the set of propositions that are 
    true in $\gamma$ and in $q = \mu(\gamma)$ is the same,
    it must also be satisfied in each state $q\in \lambda$.
    Viceversa, let $\lambda$ be a winning computation for the 
    users and $(\rho, \delta)[\lambda]$ the associated play. 
    By construction, $i-$th element in $\delta$ is associated 
    with the $i-$th state in $\pi(\lambda)$ for each $i\in \mathbb{N}$, 
    hence if $\phi$ is verified in each state in $\lambda$, it 
    must be satisfied in each cut in $\delta$. 
    Since all the cuts compatible with $\delta$ also belong to it 
    by construction, then $(\rho, \delta)[\lambda]$ is winning 
    for the users.

    \item $\psi\equiv\langle\langle users\rangle\rangle U (\phi_1, \phi_2)$, 
    $(\rho, \delta)$ be a winning play for the users, and $\lambda\in \Lambda(\rho, \delta)$. 
    By construction, the sequence of states in $\lambda$ is 
    associated with a maximal refinement $\delta'$ of $\delta$. 
    By definition, there must be a cut $\gamma$ compatible with 
    $\delta'$ in which $\phi_2$ is verified. By definition of 
    maximal refinement, $\gamma\in \delta'$, therefore there is 
    a state $q\in \lambda$ in which $\phi_2$ is verified. 
    For each cut $\gamma'\in \delta' : \gamma' < \gamma$, 
    $\phi_1$ is satisfied. By construction, each state $q'\in\lambda$ 
    preceding $q$ is associated with one of these cuts $\gamma'$, 
    hence in each of them $\phi_1$ is satisfied. 
    
    Viceversa, let $\lambda$ be an infinite fair computation 
    winning for the users, and $(\rho, \delta)$ the associated 
    play on the unfolding.
    Let $q_i$ be the state in which $\phi_2$ 
    is true for the first time; all the states $q_j\in \lambda : 
    i < j$ satisfy $\phi_1$. By construction, each cut in 
    $\delta$ can be associated with a state in $\lambda$, and the 
    order in which the cuts appear in $\delta$ is the same in 
    which the associated states are in $\lambda$, hence there must 
    be a cut $\gamma\in\delta$ associated with the state $q_i$, in 
    which $\phi_2$ is verified and such that for all the cuts 
    $\gamma'\in\delta : \gamma'<\gamma$, $\phi_1$ is verified in 
    $\gamma'$. By construction, $\delta$ cannot be refined further, 
    hence the previous fact is enough to state that $(\rho, \delta)$ 
    is winning for the users on the net.
\end{itemize}
\qed
\end{proof}
\subsection{Relation between winning strategies in the two models} 
\label{sec:relstrat}
We now restrict ourselves to the case of memoryless strategies, 
formulas without $X$ operator and in which in all the quantifiers 
there is the set of all the users. 
\begin{theorem}
Let $\psi$ be an ATL formula as described above, $\Sigma = (P, T, F, q_0, \alpha)$ be a distributed net system, $\unf(\Sigma)$ its unfolding, and $G_\Sigma$ the associated game 
structure. 
A set of memoryless winning strategies exists on $\unf(\Sigma)$ 
for $\psi$ iff a set of memoryless winning strategies exists on $G_\Sigma$ from $q_0$ for $\psi$.
\end{theorem}
%
%
\begin{proof}
As first step, we show that if there is a set of memoryless 
winning strategies on $\unf(\Sigma)$ for $\psi$, then there is a set of 
memoryless winning strategies for the users from $q_0$ 
on $G_\Sigma$.

Let $A$ be the set of users on the net, and  
$F'_A =\{f'_a:Q\rightarrow 2^{T_a}, a\in A\}$ 
a set of memoryless winning strategies, one for each user on $\unf(\Sigma)$, where $Q$ is the set of reachable markings of $\Sigma$.
We define a set of strategies $F_A$ for the users on the 
concurrent game structure in this way:
%
\begin{enumerate}
    \item If $f'_a(q) \neq \emptyset$, and $t$ 
    a transition $t\in f'_a(q)$ arbitrarily chosen, 
    then $f_a(q) = t$.
    \item $f_a(q) = \emptyset$ otherwise.
\end{enumerate} 
We prove that such a strategy is winning in $G_\Sigma$. 

Let \emph{out}$(q_0, F_A)$ be the set of fair computations starting 
from $q_0$ that the users enforce when they follow the 
strategies in $F_A$.
For any $\lambda\in$\emph{out}$(q_0, F_A)$, we show that the play 
$(\rho, \delta)[\lambda]$ satisfies the condition of Definition 
\ref{def:cplay}, and therefore it is consistent with the strategies 
in $F_A'$: (1) let $e\in E_a\cap \rho$ be an event controllable by user $a$. 
By construction, $\delta$ is a maximal refinement, 
hence there must be two cuts $\gamma_j$ and $\gamma_{j+1}$ such 
that $e$ is the only event between them. By construction, if 
$e$ was added to the run after $\gamma_j$, there must be a state 
$q\in \lambda : \mu(\gamma_j) = q$ such that $f_a(q) = \mu(e)$. 
By construction of the strategy, $\mu(e)\in f'_a(q)$.
(2) By contradiction, assume that there is a user $a$ finally 
postponed; then there is a cut $\gamma$ such that for each cut 
$\gamma_j > \gamma$, $f'_a(\mu(\gamma_j))\neq \emptyset$. By construction, 
there must be a state $q_i\in\lambda : \mu(\gamma) = q_i$ such that 
$f_a(q_j) \neq \emptyset$ for each $q_j, j > i$; then $\lambda$ 
cannot be fair with respect to $\langle k+2, c_{a}\rangle$. 

This run is a play, and it is consistent with the 
strategy by construction, hence it is a winning play for the 
users and it respects the property expressed by $\psi$. 
Hence also $\lambda$ satisfies it on the concurrent game structure.

As second step, we want to show that if $F_A$ is a set of memoryless
winning strategies for the users on $G_\Sigma$ 
with initial state $q_0$, 
then there is a set of memoryless winning strategies for the users
on $\unf(\Sigma)$.

Let $F_A = \{ f_a: a \in A \}$ be a set of winning strategies
on $G_\Sigma$.
We construct a family of strategies on $\unf(\Sigma)$
by putting for every 
user $a$ and for every state $q\in Q$, $f'_a(q) = \{f_a(q)\}$.
We want to show that $F'_A = \{ f'_a: a \in A \}$ is a set of
winning strategies for the users.

Let $(\rho, \delta)$ be a play consistent with $F'_A$.
By contradiction, we assume that the formula 
$\psi$ is not satisfied on $\unf(\Sigma)$, 
and that the play $(\rho, \delta)$ testifies it. 
Then there must be a maximal refinement $\delta'$ of $\delta$,
such that $(\rho, \delta')$ does not satisfy the formula. 
$\delta'$ is associated with an infinite sequence of states 
$\lambda$ on $G_\Sigma$. 
If $\lambda$ is a fair computation on $G_\Sigma$, 
then it is also a computation consistent with the strategies: 
every time that there is a user's move, it occurs in 
a state corresponding to a marking in which that 
transition is enabled. An event associated to this transition must be chosen in 
$\unf(\Sigma)$ in a cut corresponding to that marking, 
and therefore it is chosen also by the state corresponding 
to that marking in the concurrent game structure by construction. 
If $\lambda$ is not fair, we know by Proposition \ref{prop:ppn2cgs} 
that we can construct a fair computation $\lambda'$ 
such that $\pi(\lambda) = \pi(\lambda')$. 
In addition, from the proof of Proposition \ref{prop:ppn2cgs} 
we know that the only reason why $\lambda$ can be unfair 
is that there are users that are finally neglected 
during the play. This cannot happen to the users that have 
their strategy finally not-empty, otherwise the play on 
the net would not be consistent with their strategy 
(by definition the strategies select only enabled events).
Then, for every user $a$ neglected by the scheduler, there 
is an infinite number of states in the sequence in which 
$f_a$ selects only the transition that keeps the system 
in the same state. We add a copy of these states in the 
position coming immediately after them. 
This computation $\lambda'$  is fair with respect to 
the user, therefore it is a play consistent with the 
strategy on the concurrent game structure and by hypothesis is winning. 
By construction $\pi(\lambda) = \pi(\lambda')$, hence 
Lemma \ref{prop: equiv} guarantees that also 
$\lambda$ respect the ATL formula.
\qed
\end{proof}
%
%
%
%
We do not know if the equivalence holds also for 
full memory strategies. 
We can prove that in case of winning strategies on the 
distributed net system, there is a strategy on the turn-based 
asynchronous game structure, but we could not prove 
anything for the opposite direction. 
\begin{proposition}
Let $F'_A =\{f'_a:\Gamma\rightarrow 2^{T_a}: a\in \{1,...,k\}\}$ be
a set of winning strategies for the users on $\unf(\Sigma)$ for the 
ATL formula $\langle\langle A\rangle\rangle \eta$. 
Then there is a winning strategy for the users from $q_0$ 
on $G_\Sigma$ for the same formula.
\label{prop:stcpn2cgs}
\end{proposition}
\begin{proof}
Let $\gamma$ be a cut on the unfolding. The set of events that 
occurred from the initial cut to $\gamma$ is uniquely determined. 
We consider all the sequentializations of of markings generated by 
firing the transitions in this set respecting the partial order 
imposed by the unfolding. By construction, 
each of these sequences can be reproduced on $G_\Sigma$ 
starting from $q_0$ on the $\cgs$. 
We denote with $\Lambda(\gamma)$ the set of finite sequences 
obtained in this way. 

For each player $a\in \{1,...,k\}$ on $G_\Sigma$, we define a 
strategy $f_a$ in this way: 
\begin{enumerate}
    \item If $f'_a(\gamma) \neq \emptyset$, 
    then $f_a(\lambda) = t$ for each $\lambda\in \Lambda(\gamma)$, where $t\in f'_a(\gamma)$.
    \item For each pair $\lambda, \lambda'$ such that 
    $\pi(\lambda) = \pi(\lambda')$, 
    $f_a(\lambda) = f_a(\lambda')$.
    \item $f_a(\lambda) = \emptyset$ otherwise.
\end{enumerate} 
We prove that such a strategy is winning in $G_\Sigma$. 
Let \emph{out}$(q_0, F_A)$ the set of fair computations starting 
from $q_0$ that the players in $A$ enforce when they follow the 
strategies in $F_A$.
For any $\lambda\in$\emph{out}$(q_0, F_A)$, we consider the play $(\rho, \delta)[\lambda]$ on the unfolding of the distributed net 
system.
We have to show that it is consistent with the 
strategies: 
the first point of Definition \ref{def:cplay} is 
satisfied by construction; 
by contradiction, we assume that the second point is not 
respected. By construction, if there is a user $a$ for 
which $f'_a(\gamma)\neq \emptyset$ for every 
$\gamma > \gamma' : \gamma'\in \rho$, then there is also 
an index $l$ such that for every $m > l$ for every 
prefix $\lambda_m$ of $\lambda$, $f_a(\lambda)\neq \emptyset$. 
If no event belonging to user $a$ occurs in $\lambda$ 
after the prefix $\lambda_l$, then player $a$ has been 
finally neglected by the scheduler, and $\lambda$ cannot 
be a fair computation.
Hence $(\rho, \delta)[\lambda]$ is a consistent play, and 
by hypothesis it is  winning for the 
users. Then by Lemma \ref{lem:eqwc} $\lambda$ is also winning for 
the users on $G_\Sigma$.
\qed\end{proof}
%
\section{Related works and conclusions}
\label{sec:last}
In this work we have defined a game for a group of users and an 
environment on the unfolding of distributed net systems,
by adapting the idea of game on the unfolding defined in \cite{ABP19} 
for two players, a sequential user and the environment.

While in \cite{ABP19} the goal of the user was to force the 
occurrence of a target transition,
here we specify the possible goals of the users with the 
temporal logic ATL, initially introduced in \cite{AHK02}
on concurrent game structures.
With respect to \cite{ABP19}, we have also modified some of the
fairness constraints in order to 
put it in relation with the game developed in \cite{AHK02},
and we have shown that in the case of memoryless strategies, 
the users have a winning strategy on the unfolding if, and only if, 
they have a winning strategy on a derived concurrent game structure.

We now discuss related works and some critical issues in our
approach.

After its introduction in \cite{AHK02}, ATL was studied and   
expanded by several authors. In particular, \cite{BPQR15}
deals with ATL with partial observability and a particular set 
of fairness constraints, defined as \emph{unconditional}. 
For this case, a model-checker of ATL is also available 
(\cite{LQR17}). 

Temporal logics are usually model checked on Kripke 
models, but some authors also defined procedures to verify formulas 
directly on Petri nets: \cite{EH00} and \cite{DLX19} provide 
model checking algorithms, respectively for LTL and CTL, based on 
the unfolding of Petri nets; the tool described in \cite{Wolf18a} 
exploits algorithms previously developed in order to 
verify LTL, CTL, and other properties on Petri nets.

Finally, in the last years, other authors defined games on Petri nets 
in order to study reachability or safety properties on partially 
controllable concurrent systems. In the game developed in 
\cite{FO17}, the players are the tokens of the net, and the 
goal is the avoidance of a certain place. In \cite{KLMS20}, the 
goal is the reachability of a target marking, and the moves of 
the players are strictly alternating on the net.
Closer to the game introduced here, in \cite{BKP17} the authors 
present a two-player game for the reachability of a target 
transition developed on the case graph of the net; 
in this game, every player can fire enabled 
transitions without any alternating rule.

\vspace{\baselineskip}

With respect to critical issues in our approach, note that in Def.~\ref{def:cplay} we require that 
each event controllable by a user occurring in a consistent 
play is immediately preceded and followed by a cut; in this way 
we are able to specify the exact moment in which the event 
occurred in the play, even if other events may be concurrent 
with it.

For many goals expressed with ATL this does not create any 
problem, however there are cases in which these constraints 
lead to unrealistic scenarios, as illustrated in Ex. \ref{ex:concpr}. 
\begin{figure}
    \centering
    \includegraphics[width=0.65\textwidth]{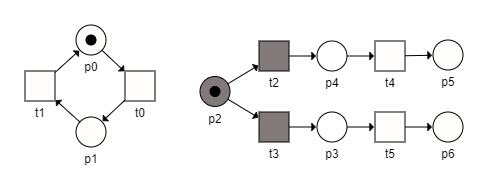}
    \caption{A distributed net with only one user and two concurrent components}
    \label{fig:concpr}
\end{figure}
\begin{example}
\label{ex:concpr}
The distributed net system in Fig.~\ref{fig:concpr} 
has two locations: the environment, represented in white and a 
user $u$, in grey. Let $\psi$ be the ATL formula 
\begin{equation}
    \psi \equiv \langle\langle u \rangle\rangle F((p_0 \wedge p_3) \lor (p_1 \wedge p_4)).
\end{equation}
Consider the strategy $f:Q\rightarrow 2^{T_u}$ defined in this way: 
$f(\{p_0, p_2\}) = \{t_3\}$, $f(\{p_1, p_2\}) = \{t_2\}$. 
In each play consistent with the strategy there must be either 
an occurrence of $t_2$ or an occurrence of $t_3$; if not, $u$ 
would be finally postponed, since, when $p_2$ is marked, his 
strategy is never empty. 
Then the user $u$ has a winning strategy, because 
whenever a transition in $T_u$ fires, the formula will be satisfied. 
This means that there must be a point of the play in which the user is able to 
make his move before the environment. 
In a concurrent system this is unrealistic: 
the two disjoint component of the net system (the left component  formed by 
$t_0$ and $t_1$ and the right formed by $t_2, t_3, t_4$ and $t_5$)
behave independently, and there is no reason to impose 
interleaving between the left component and the transitions of the 
user.

In addition to that, if we consider the formula:
\begin{equation}
    \langle\langle u \rangle\rangle F(p_0 \wedge p_3),
\end{equation}
the user does not have any winning strategy; the strategy with $f(\{p_0, p_2\}) = \{t_3\}$, 
$f(q) = \emptyset$ for every $q \neq \{p_0, p_2\}$ is not winning, because the move of 
the user is not finally postponed, since in the marking $\{p_1, p_2\}$ the strategy 
is empty, and $t_0t_1$ repeated infinitely often would be a valid execution for a consistent play. Filling the strategy by putting $f(\{p_1, p_2\}) = \{t_3\}$ would not make 
the strategy winning, because, even if $t_3$ must be in the run, the user cannot 
be sure that $\{p_3, p_0\}$ is reached. A valid execution would be $t_0, t_3, t_5, t_1...$, 
that never crosses the desired places together; the user cannot 
impose any order to the transitions of the environment.
\end{example}
Starting from these considerations, we plan to study the 
existence of a strategy directly on structures that explicitly 
represent concurrency, such as the unfolding of a net, 
without forcing sequential executions; 
on these structures, we will also study formulas with the `next' 
operator.
Another goal is to consider partial observability of the users, dropping 
the assumption that the users always know the global state of the 
system. 
This will help to have a more realistic representation of 
interactions on concurrent systems.

Finally, we plan to investigate further whether ATL can be 
successfully used to study properties of concurrent systems, 
and possibly consider an extension of it to define properties 
on the nets.
%
\bibliographystyle{plain}
\bibliography{atl}
\end{document}